\documentclass[10pt]{article}
\usepackage{graphicx}
\usepackage{multirow}
\usepackage{epsfig,caption}

\usepackage[T1]{fontenc}
\usepackage[sc]{mathpazo}
\usepackage{amsmath}
\usepackage{amssymb}
\usepackage{enumerate}
\usepackage{amsthm}
\usepackage{amsfonts,mathrsfs}

\usepackage{hyperref}
\hypersetup{pdfpagemode=UseNone}

\newcommand{\tinyspace}{\mspace{1mu}}

\newcommand{\op}[1]{\operatorname{#1}}

\newcommand{\abs}[1]{\left\lvert\tinyspace #1 \tinyspace\right\rvert}

\newcommand{\norm}[1]{\left\lVert\tinyspace #1 \tinyspace\right\rVert}

\renewcommand{\t}{{\scriptscriptstyle\mathsf{T}}}

\newcommand{\setft}[1]{\mathrm{#1}}

\newcommand{\density}[1]{\setft{D}\left(#1\right)}

\renewcommand{\vec}{\op{vec}}
\newcommand{\im}{\op{im}}
\newcommand{\rank}{\op{rank}}

\def\complex{\mathbb{C}}

\def\I{\mathbb{1}}

\newenvironment{mylist}[1]{\begin{list}{}{
    \setlength{\leftmargin}{#1}
    \setlength{\rightmargin}{0mm}
    \setlength{\labelsep}{2mm}
    \setlength{\labelwidth}{8mm}
    \setlength{\itemsep}{0mm}}}
    {\end{list}}

\def\ot{\otimes}

\newcommand{\inner}[2]{\langle #1 , #2\rangle}
\newcommand{\iinner}[2]{\langle #1 | #2\rangle}
\newcommand{\out}[2]{| #1\rangle\langle #2 |}

\newcommand{\Innerm}[3]{\left\langle #1 \left| #2 \right| #3 \right\rangle}

\newcommand{\pa}[1]{(#1)}
\newcommand{\Pa}[1]{\left(#1\right)}

\newcommand{\Br}[1]{\left[#1\right]}
\newcommand{\set}[1]{\{#1\}}
\newcommand{\Set}[1]{\left\{#1\right\}}

\newcommand{\bra}[1]{\langle#1|}

\newcommand{\ket}[1]{|#1\rangle}

\DeclareMathOperator{\trace}{Tr}
\newcommand{\ptr}[2]{\trace_{#1}\pa{#2}}
\newcommand{\Ptr}[2]{\trace_{#1}\Pa{#2}}
\newcommand{\tr}[1]{\ptr{}{#1}}
\newcommand{\Tr}[1]{\Ptr{}{#1}}

\def\cH{\mathcal{H}}

\def\rA{\mathrm{A}}
\def\rF{\mathrm{F}}\def\rG{\mathrm{G}}

\def\rR{\mathrm{R}}\def\rS{\mathrm{S}}
\def\rU{\mathrm{U}}

\newtheorem{thrm}{Theorem}[section]

\newtheorem{prop}[thrm]{Proposition}

\theoremstyle{definition}

\numberwithin{equation}{section}

\newcounter{questionnumber}

\usepackage{color}

\newcommand{\proj}[1]{| #1\rangle\!\langle #1 |}
\newcommand{\ketbra}[2]{|#1\rangle\!\langle#2|}

\def\r{\rho}
\def\s{\sigma}
\def\dg{\dagger}
\def\ps{\psi}

\newcommand{\jmp}{J. Math. Phys.~}
\newcommand{\jpa}{J. Phys. A~}

\newcommand{\njp}{New. J. Phys.~}

\newcommand{\prl}{Phys. Rev. Lett.~}
\newcommand{\pra}{Phys. Rev. A~}
\newcommand{\pla}{Phys. Lett. A~}
\newcommand{\rmp}{Rev. Math. Phys.~}

\begin{document}

\title{Fidelity between one bipartite quantum state and another undergoing local unitary dynamics}

\author{Lin Zhang$^1$\footnote{E-mail: godyalin@163.com;
linyz@zju.edu.cn}\ , Lin Chen$^{2,3}$\footnote{E-mail:
linchen@buaa.edu.cn (corresponding author)}\ , Kaifeng Bu$^4$\footnote{E-mail:bkf@zju.edu.cn}\\
  {\small $^1$\it Institute of Mathematics, Hangzhou Dianzi University, Hangzhou 310018, PR~China}\\
  {\small $^2$\it School of Mathematics and Systems Science, Beihang University, Beijing 100191, PR~China}\\
  {\small $^3$\it International Research Institute for Multidisciplinary Science, Beihang University, Beijing 100191, PR~China}
   \\
  {\small $^4$\it Department of Mathematics, Zhejiang University, Hangzhou 310027, PR~China}}
\date{}
\maketitle
\begin{abstract}
The fidelity and local unitary transformation are two widely useful
notions in quantum physics. We study two constrained optimization
problems in terms of the maximal and minimal fidelity between two
bipartite quantum states undergoing local unitary dynamics. The
problems are related to the geometric measure of entanglement and
the distillability problem. We show that the problems can be reduced
to semi-definite programming optimization problems. We give
close-form formulae of the fidelity when the two states are both
pure states, or a pure product state and the Werner state. We
explain from the point of view of local unitary actions that why the
entanglement in Werner states is hard to accessible. For general
mixed states, we give upper and lower bounds of the fidelity using
tools such as affine fidelity, channels and relative entropy from
information theory. We also investigate the power of local unitaries, and the
equivalence of the two optimization problems.
\end{abstract}

\section{Introduction}

Finding suitable quantities for characterizing the correlations in a
bipartite or multipartite quantum state has been an important
problem in quantum information theory. Three well-known quantities
are entanglement, fidelity and mutual information \cite{Watrous}.
Investigating the quantities under unitary dynamics has various
physical applications. The local evolution of free entangled states
into bound entangled or separable states in finite time presents the
phenomenon of sudden death of distillability. In the phenomenon, the
fidelity was used to evaluate how close the evolved state is close
to the initial state \cite{scz09}. Next, finding out the local
unitary orbits of quantum states characterizes their properties for
various quantum-information tasks, and it is also mathematically
operational  \cite{SK-jpa}-\cite{GW-prl}. By searching for the
maximally and minimally correlated states on a unitary orbit, the
authors in \cite{Jevtic} quantified the amount of correlations in
terms of the quantum mutual information. The correlations in a
multipartite state within the construction of unitary orbits have
been also examined \cite{Modi}. These applications originate from
the fact that the unitary dynamics influences the  interaction of
quantum systems. It is thus a widely concerned question to
characterize how heavy the influence can be under certain metric
such as the fidelity. The latter has been used to evaluate the
entangled photon pairs obtained by experimental heralded generation
\cite{bcz10}, and the unitary gates of experimentally implementing
quantum error correction \cite{sbm11}. In contrast to the global
unitary dynamics which involves nonlocal correlation, the local
unitary action can be locally performed and does not change the
properties of quantum states. Because of the easy accessibility in
mathematics, the global unitary dynamics has been studied a lot
\cite{Zhang}. In contrast, much less is known about the local
unitary dynamics.

In this paper, we study the maximal and minimal fidelity between two
bipartite quantum states, one of which undergoes arbitrary local
unitary dynamics. To be more specific, let $\r$ and
$\s$ be two bipartite states acting on the Hilbert space
$\cH_1\ot\cH_2$ of dimensions $\dim \cH_i= d_i$, $i=1,2$. Let
$\rU(\cH_1)$ be the unitary group on $\cH_1$. We propose two
constrained optimization problems as computing the functionals
\begin{eqnarray} \label{eq:max}
\rG_{\max}(\r,\s):= \max_{U_i\in\rU(\cH_i):i=1,2}
\rF(\r,(U_1\ot U_2)\s(U_1\ot U_2)^\dagger)
\end{eqnarray}
and
\begin{eqnarray} \label{eq:min}
\rG_{\min}(\r,\s):= \min_{U_i\in\rU(\cH_i):i=1,2}
\rF(\r,(U_1\ot U_2)\s(U_1\ot U_2)^\dagger)
\end{eqnarray}
where $\rF(\rho,\sigma):=\Tr{\sqrt{\sqrt{\rho}\sigma\sqrt{\rho}}}$
is the fidelity between any two semidefinite positive matrices $\rho$ and $\sigma$.
Because of the symmetric property of fidelity, the two functionals
are unchanged under the exchange of arguments $\r$ and
$\s$. Note that if the local unitary action is replaced by
global unitary action, then the problems have been analytically
solved in \cite{Zhang}.

Intuitively, the functionals respectively stand for the maximal and
minimal distance that local unitary can create between quantum
states. The solution to the optimization problems exists because the
unitary group is a compact Lie group. We will show that they can
indeed be reduced to the well-known semidefinite programming (SDP)
problems. So we may efficiently compute the functionals for many
states. Then we derive the close-form formulaes to the functionals
when $\r$ and $\s$ are both pure states, or a pure product
state and the Werner state. We show that in contrast with the
separable Werner state, the entangled Werner state of $d>3$ is
closer to the set of pure separable states under local unitary
dynamics. In this context, the distillability of two-qubit and
two-qutrit Werner states may be distinguished by comparing their
$\rG_{\max}$. For general mixed states, we derive the upper and
lower bounds of the functionals in terms of the monotonicity of
fidelity, quantum channel, the affine fidelity, the integral over
the unitary group via Haar measure. We also investigate how
local unitaries influence the commutativity of quantum states, as well as the equivalence of the two optimization problems.

Our results straightforwardly make progress towards the following
quantum-information problems. First, $\rG_{\max}(\r,\s)$ reduces to
the geometric measure of entanglement (GME) when $\r$ or $\s$ is a
pure product state \cite{shimony95,wg03}. Mathematically the GME of
a quantum state $\r$ is defined as $\max_{\ps}\bra{\ps}\r\ket{\ps}$
where $\ket{\ps}$ is a product state. It is known that the GME of a
bipartite state measures the closest distance between this state and
separable states. It coincides with the intuitive interpretation of
the functionals. The GME is a multipartite entanglement measure and
has been extensively studied recently \cite{zch10,cah14}. The GME
also applies to the construction of initial states for Grover
algorithm \cite{bno02,ssb04}, the discrimination of quantum states
under local operations and classical communications (LOCC)
\cite{mmv07}, and one-way quantum computation \cite{zch10}. For a
review of GME we refer the readers to \cite{cah14}. Recall that the
fully entangled fraction (FEF) for any bipartite state $\rho$ in a $d\ot
d$ system is defined as the maximal overlap with maximally entangled
pure states,
$$
\max_{U,V~\text{unitaries}} \Innerm{\Omega}{(U\ot V)\rho (U\ot
V)^\dagger}{\Omega},
$$
where $\ket{\Omega}=\frac1{\sqrt{d}}\sum^{d-1}_{j=0}\ket{jj}$ is the maximally entangled state. Then,
$\rG_{\max}(\r,\s)$ is the square root of the FEF when one state of $\r$ and $\s$ is a maximally entangled state for
$d_1=d_2$ \cite{Zhao}. In this case, the other state of $\r$ and $\s$ can be any mixed state. The FEF works as the fidelity of optimal
teleportation, and thus has experimental significance
\cite{Bennett}. The close-form for the FEF in a two-qubit system is
derived analytically by using the method of Lagrange multiplier
\cite{GEJ}. Second, $\rG_{\min}(\r,\s)$ is related to the famous
distillability problem in entanglement theory. The latter is related
to the additivity property of distillable entanglement and the
activation of bound entanglement \cite{sst03}. It is known that a
bipartite state $\r$ is distillable if and only if there exists a
positive integer $n$ and a Schmidt-rank two pure state $\ket{\psi}$,
such that $\bra{\psi}(\r^{\ot n})^\Gamma\ket{\psi}<0$
\cite{dss00,peres96}. Our optimization problems imply that $\r$ is
distillable if and only if $\min_{\lambda\in(0,1)}
\rG^2_{\min}(\ket{\phi_\lambda},(\r^{\ot n})^\Gamma+x \I)<x$, where
$\ket{\phi_\lambda}=\sqrt{\lambda}\ket{00}+\sqrt{1-\lambda}\ket{11}$
and $x$ is a positive number such that the second argument is
positive semi-definite. We stress that the difficulty of the
distillability problem mostly arises from the local unitary orbits
involved in the optimization problems above. The distillability
problem has turned out to be hard, and a review of recent progress
can be found in \cite{cd11}. All these problems are thus well
motivated by the findings in this paper.

The rest of the paper is organized as follows. In
Sec.~\ref{sect:analytical-formula}, we show that the computation of
the two functionals $\rG_{\max}$ and $\rG_{\min}$ can be reduced to
the SDP problem. Then we derive the close-form formulae of
functionals when $\r$ and $\s$ are both pure states, or a
pure product state and the Werner state. We also point out a
potential connection between the distillability problem and our
optimization problem for $\rG_{\max}$. Next, several connections,
upper and lower bounds on the functionals are computed in
Sec.~\ref{sect:upper-lower-bounds}. We discuss in Sec.~\ref{sect:examples}, and
conclude in Sec.~\ref{sect:conclusion}.

\section{SDP and analytical formula of functionals}\label{sect:analytical-formula}

We see that $\rF(\rho,(U\ot V)\sigma(U\ot V)^\dagger)$ is a
continuous function over local unitary groups
$\rU(\cH_1)\ot\rU(\cH_2)$. Since $\rU(\cH_1)$ and $\rU(\cH_2)$ are
compact Lie groups, it follows that there exists $U_i,
V_i\in\rU(\cH_i)(i=1,2)$ such that
$$
\rG_{\max} (\rho,\sigma) = \rF(\rho,(U_1\ot U_2)\sigma(U_1\ot
U_2)^\dagger)~~\text{and}~~\rG_{\min} (\rho,\sigma) =
\rF(\rho,(V_1\ot V_2)\sigma(V_1\ot V_2)^\dagger).
$$
Denote $\widehat\s:=(U_1\ot U_2)\s(U_1\ot U_2)^\dagger$ and
$\widetilde\sigma = (V_1\ot V_2)\sigma(V_1\ot V_2)^\dagger$. Thus
$$
\rG_{\max} (\rho,\sigma) = \rF(\rho,\widehat\sigma)
~~\text{and}~~\rG_{\min} (\rho,\sigma) = \rF(\rho,\widetilde
\sigma).
$$
The SDP has been extensively used to treat the distillability
problem \cite{rains01}, the separability problem \cite{dps02}, the
quantification of entanglement \cite{brandao05} and so on
\cite{hhh09}. The SDP for fidelity between two states is obtained by
Watrous \cite{Watrous}. We show that our problems of computing
$G_{\max}$ and $G_{\min}$ can be reduced to the SDP optimization
problem \cite{JW,Kill}, as a primal problem below. Let
$\tau=\widehat\sigma$ or $\widetilde \sigma$. Then
\begin{eqnarray}
\text{maximize:}&&~~~\frac12\Pa{\Tr{X}+\Tr{X^\dagger}},\\
\text{subject to:}&&~~~\Br{
                     \begin{array}{cc}
                       \r & X \\
                       X^\dagger & \tau \\
                     \end{array}
                   }\geqslant 0,
\end{eqnarray}
where $X$ is a operator of order $d_1d_2$. Under the above
constraint the optimal value of $\frac12\Pa{\Tr{X}+\Tr{X^\dagger}}$
is fidelity $\rF(\sigma,\tau)$. Its dual problem is
\begin{eqnarray}
\text{minimize:}&&~~~\frac12\Pa{\inner{\r}{Y}+\inner{\tau}{Z}},\\
\text{subject to:}&&~~~\Br{
                     \begin{array}{cc}
                       Y & -\I \\
                       -\I & Z \\
                     \end{array}
                   }\geqslant 0,
\end{eqnarray}
where $Y,Z$ are Hermitian operators. So we can numerically solve the
optimization problems for many states with high efficiency. On the
other hand, we can  analytically solve the problems for pure states.

\begin{thrm}
\label{thm:rho12} Let $\r=\out{\Phi_{12}}{\Phi_{12}}$ and
$\s=\out{\Psi_{12}}{\Psi_{12}}$, where the spectra of reduced
density operators $\r_1$ and $\s_1$ are $\set{a_1\geqslant
\cdots\geqslant a_N\geqslant0}$ and $\set{b_1\geqslant
\cdots\geqslant b_N\geqslant0}$, respectively, and $d=d_1=d_2$. Then
$\rG_{\max}(\r,\s) =\sum^d_{j=1} \sqrt{a_jb_j}$ and
$\rG_{\min}(\r,\s) = 0$.
\end{thrm}
\begin{proof}
There are two $d_2\times d_1$ matrices $A,B$ such that
$\ket{\Phi_{12}}=\vec(A)$ and $\ket{\Psi_{12}}=\vec(B)$
\cite{Watrous}. Then $(U_1\ot U_2)\ket{\Psi_{12}} =
\vec(U_1BU^\t_2)$ implies $\rF(\r,(U_1\ot
U_2)\s(U_1\ot U_2)^\dagger) =\abs{\Tr{A^\dagger
U_1BU^\t_2}}. $ Since $\rho_1=AA^\dagger$ and $\sigma_1=BB^\dagger$,
the first assertion follows from the fact $
\max_{U,V}\abs{\Tr{XUYV}} = \sum^N_{k=1} s_k(X)s_k(Y) $
\cite{Bhatia}. The second assertion is equivalent to the fact that
the product states in the Schmidt decomposition of $\ket{\Phi_{12}}$
can be by local unitaries converted into states orthogonal to those
of $\ket{\Psi_{12}}$.
\end{proof}

When $\ket{\Phi_{12}}$ is from a maximally entangled basis, the
proof for $\rG_{\min}(\r,\s) = 0$ is similar to the technique for
the known quantum super-dense coding. On the other hand, a
particular case of this theorem has been used to construct a family
of entanglement witnesses \cite{PM}. Next if $\r$ is a pure product
state, then $\rG^2_{\max}(\r,\s)$ amounts to the $S(1)$-norm
$\norm{\s}_{S(1)}$, which is lower bounded by the $(d_1+d_2-1)$-th
eigenvalue in increasing order of $\s$ \cite{Johnston}, where the
$S(k)$-norm of bipartite operator $X$ on $\cH_1\ot\cH_2$ is defined
as
\begin{eqnarray*}
\norm{X}_{S(k)} := \sup_{\ket{u},\ket{v}\in\cH_1\ot\cH_2}
\Set{\abs{\Innerm{u}{X}{v}}:
\rS\rR(\ket{u}),\rS\rR(\ket{v})\leqslant k}.
\end{eqnarray*}
Here $\rS\rR(\ket{w})$ stands for the Schmidt-rank of pure bipartite
state $\ket{w}\in\cH_1\ot\cH_2$, i.e. the number of nonvanishing
coefficients in Schmidt decomposition of $\ket{w}$.

The problem can be analytically solved for the Werner state $ \s(t)
= \frac1{d(d-t)} ( \I_d\ot\I_d - t\sum_{i,j=1}^d\out{ij}{ji} ), $
$t\in[-1,1]$. Indeed \cite{Kribs} implies
\begin{eqnarray}
\rG_{\max}(\r,\s)= \sqrt{\norm{\s}_{S(1)}} =
\sqrt{\frac{1+\abs{\min(t,0)}}{d(d-t)}}.
\end{eqnarray}
One can easily show that the minimum of this functional over $t$ is
achievable when $t=0$. That is, the Werner state becomes the
maximally mixed state, which is at the center of the set of
separable states. Next, the maximum of the functional may be reached
at two points, $t=-1$ and $1$ as plotted in Figure \ref{fig}. The
two points respectively correspond to a separable Werner state and
the most entangled Werner state.  Figure \ref{fig} implies that in
contrast with the separable Werner state, the entangled Werner state
of $d>3$ is closer to the set of pure separable states under local
unitary dynamics. That is, the entanglement of Werner states might
be a more unaccessible quantum correlation than the separability. It
is known that two-qubit entangled states are distillable, and it is
conjectured that two-qutrit entangled Werner states may be
non-distillable \cite{vd06}. Our inequalities imply that the
distillability of Werner states of $d=2$ and $d>2$ may be
essentially distinguished by their distance to the pure separable
states under local unitary dynamics.

\begin{figure}[ht]
\centering\scalebox{0.55}{\includegraphics{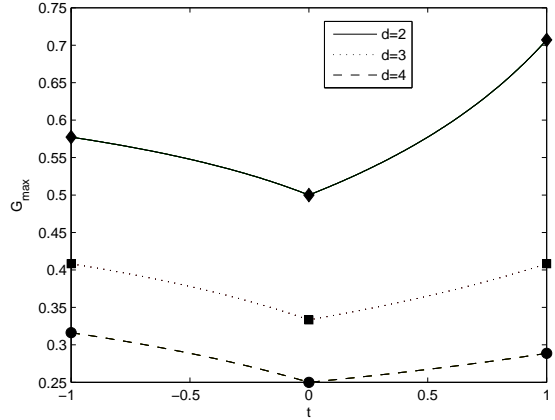}}
\captionsetup{font=small,width=0.55\textwidth}
\caption{The figure shows the maximum distance between the Werner
state and the set of pure product states in terms of the functional
$\rG_{\max}$. The leftmost value is smaller than the rightmost value
for $d=2$, and bigger than the rightmost value for $d=4$. The two
values are equal for $d=3$.} \label{fig}
\end{figure}

In Appendix \ref{app:iso}, we also have computed $\rG_{\max}$ for
the pure product state and the isotropic state. In spite of these
results, finding the analytical solution to the optimization
problems is unlikely because local unitary actions are much more
involved than global unitary action $U$. Indeed, the extremal values
of $\rF(\rho,U\sigma U^\dagger)$ are determined by those global
unitary such that the commutator $[\rho,U\sigma U^\dagger]=0$
\cite{Zhang}. For our purpose we need to replace the global unitary
action by local unitary $U_1\ot U_2$. We do not know whether there
exist $U_1$ and $U_2$ such that $\rF(\r,(U_1\ot
U_2)\s (U_1\ot U_2)^\dagger)$ can attain its extremal
values. In next section we study the upper and lower bounds of
$\rG_{\max}$ and $\rG_{\min}$, as well as their connections.

\section{Upper and lower bounds of functionals}\label{sect:upper-lower-bounds}

The optimization problems ask to find out the critical points $\arg
\rG_{\max}(\r,\s)$ and $\arg \rG_{\min}(\r,\s)$
in the local unitary group, which respectively achieve $G_{\max}$
and $\rG_{\min}$. They respectively refer to the local unitary
operator $V_1\ot V_2$ such that $\rG_{\max}(\r,\s) =
\rF(\r,(V_1\ot V_2)\s(V_1\ot V_2)^\dagger)$ and
$\rG_{\min}(\r,\s) = \rF(\r,(V_1\ot
V_2)\s(V_1\ot V_2)^\dagger)$. Using these facts, we construct
several relations between $G_{\max}$ and $\rG_{\min}$. In
Proposition \ref{pp:max+min} and \ref{pp:rho12}, we show that the
sum and difference of $\rG_{\max}$ and $\rG_{\min}$ (or their
squares) are both upper and lower bounded in terms of some functions
of their arguments such as the rank and fidelity. As a result, we
study the fidelity inequality in Proposition
\ref{prop:fidelity-run}, and the affine fidelity of global unitary
action in Proposition \ref{pp:affine}. We construct the bounds of
$\rG_{\max}$ and $\rG_{\max}$ by using the monotonicity of the
fidelity, the affine fidelity, the integral over the unitary group
via Haar measure, and the relative entropy in Subsect \ref{subsec:a}
and \ref{subsec:b}.

\begin{prop}
\label{pp:max+min} Let $\r,\s$, and $\s'={1\over
d_1d_2-1}(\I_{d_1d_2}-\s)$ be three quantum states. We have
\begin{eqnarray}\label{eq:max+min}
\rank (\r)\geqslant \rG_{\max}(\r,\s)^2 +
(d_1d_2-1)\rG_{\min}(\r,\s')^2 \geqslant 1.
\end{eqnarray}
The first equality holds if there exist two unitary matrices
$W_1,W_2$ such that conditions \eqref{1},\eqref{2} and \eqref{3}
hold. The second equality holds if and only if there exist two
unitary matrices $V_1,V_2$ such that conditions \eqref{4},\eqref{5},
and \eqref{6} hold:
\begin{enumerate}[(i)]
\item\label{1}
$W_1\ot W_2 = \arg G_{\max}(\r,\s) = \arg
G_{\min}(\r,\s')$;
\item\label{2}
$\sqrt{\r}(W_1\ot W_2)\s(W_1\ot W_2)^\dagger
\sqrt{\r}$ and $\sqrt{\r}(W_1\ot
W_2)\s'(W_1\ot W_2)^\dagger \sqrt{\r}$ both have
identical nonzero eigenvalues;
\item\label{3} $\rank \Pa{\sqrt{\r}(W_1\ot W_2)\sqrt{\s}}= \rank
\Pa{\sqrt{\r}(W_1\ot
W_2)\sqrt{\s'}}=\rank(\r)$;
\item\label{4} $V_1\ot V_2 = \arg
\rG_{\max}(\r,\s) = \arg
\rG_{\min}(\r,\s')$;
\item\label{5} $\sqrt{\r}(V_1\ot V_2)\s(V_1\ot
V_2)^\dagger\sqrt{\r}$ has rank one;
\item\label{6} $\sqrt{\r}(V_1\ot V_2)\sigma'(V_1\ot
V_2)^\dagger\sqrt{\r}$ has rank one.
\end{enumerate}
\end{prop}

The proof is given in Appendix~\ref{app:proof}. One can easily
verify that the second equality in \eqref{eq:max+min} holds when
$\r$ is pure, or $\r={1\over2}(\proj{00}+\proj{01})$ and
$(V_1\ot V_2)\s(V_1\ot V_2)^\dagger=\proj{00}$. In both
cases, computing $\rG_{\max}$ and $\rG_{\min}$ are equivalent. This
is the first connection we have between the two functionals, so it
is enough to consider only one of  them. We shall discuss the
general case in Sec.~\ref{sect:examples}. Furthermore, conditions 5
and 6 imply that $\r$ has rank at most two. If it has rank two,
then conditions 5 and 6 imply that $\s$ is pure. Thus, at least
one of $\r$ and $\s$ is pure when the second equality in
\eqref{eq:max+min} holds. Next we construct another restriction
between $\rG_{\max}$ and $\rG_{\min}$ or their squares. This is
realized based on the inequalities for the framework of
wave-particle duality \cite{Coles} and the ensembles of Holevo
quantity \cite{Roga}.

\begin{prop}\label{pp:rho12}
Let $\r,\s$, and $\s'={1\over d_1d_2-1}(\I_{d_1d_2}-\s)$ be three
quantum states. Assume that $U_1\ot U_2=\arg \rG_{\max}(\r,\s)$ and
$V_1\ot V_2=\arg \rG_{\min}(\r,\s')$. Then
\begin{eqnarray}
\rG_{\max}(\r,\s) +
\rG_{\min}(\r,\s')\leqslant\sqrt{2+2\rF(\widehat\s,\widehat\s')},
\end{eqnarray}
where $\widehat\s=(U_1\ot U_2)\s(U_1\ot U_2)^\dagger$ and
$\widehat\s'=(V_1\ot V_2)\s'(V_1\ot V_2)^\dagger$. Moreover,
\begin{eqnarray}
\abs{\rG^2_{\max}(\r,\s) - \rG^2_{\min}(\r,\s') }
\leqslant\sqrt{1-\rF^2(\widehat\s,\widehat\s')}
\end{eqnarray}
and
\begin{eqnarray}
\abs{ \rG_{\max}(\r,\s) - \rG_{\min}(\r,\s') }
\leqslant\sqrt{1-\rF^2(\widehat\s,\widehat\s')}.
\end{eqnarray}
\end{prop}

\begin{proof}
The assertion straightforwardly follow from three inequalities in
\cite{Coles,AR} and \cite{Roga}:
\begin{eqnarray}
\max_{\r} (\rF(M,\r)+\rF(N,\r)) = \sqrt{\Tr{M}+\Tr{N} + 2\rF(M,N)},
\end{eqnarray}
\begin{eqnarray}
\abs{\rF^2(\s,\r)-\rF^2(\tau,\r)} \leqslant \sqrt{1-\rF^2(\s,\tau)},
\end{eqnarray}
and
\begin{eqnarray}
\abs{\rF(\s,\r)-\rF(\tau,\r)}\leqslant \sqrt{1-\rF^2(\s,\tau)}.
\end{eqnarray}
where $M,N$ are positive semidefinite operators, and $\r,\s,\tau$
are three quantum states.
\end{proof}
The results show that the characterization of fidelity is important
for obtaining a tight bound for the functionals. We study the
characterization using an inequality for the approximation of Markov
chain property \cite{Fawzi}.
\begin{prop}\label{prop:fidelity-run}
Let $\rho$ and $\sigma$ be two quantum states on $\complex^d$, and
$\Phi$ be a quantum channel over $\complex^d$. Then
\begin{eqnarray}\label{eq:prop}
\rF(\rho,\sigma) \leqslant \sum_j \rF(M_j\rho M^\dagger_j,M_j\sigma
M^\dagger_j)\leqslant \rF(\Phi(\rho),\Phi(\sigma)),
\end{eqnarray}
where $\Phi(*)=\sum_j M_j * M^\dagger_j$ is any Kraus representation
of $\Phi$.
\end{prop}

\begin{proof}
From Lemma~B.7 in \cite{Fawzi}, we know that for an identity
resolution $\sum_j E_j=\I$,
$$
\rF(\rho,\sigma)\leqslant \sum_j \rF(E_j \rho E_j,\sigma).
$$
Since $\Phi(*)=\sum_j M_j * M^\dagger_j$ is a quantum channel,
$\sum_k M^\dagger_jM_j=\I$. Assuming $E_j=M^\dagger_jM_j$ in the
above inequality, we have
\begin{eqnarray}
\rF(\rho,\sigma)\leqslant \sum_j \rF(M_j^\dagger M_j \rho
M_j^\dagger M_j,\sigma).
\end{eqnarray}
Again, by employing the following simple fact, Lemma~B.6 in
\cite{Fawzi}, $\rF(W^\dagger\rho W,\sigma) = \rF(\rho,W\sigma
W^\dagger)$, we obtain the inequality in \eqref{eq:prop}. The other
inequality is from the concavity of fidelity.
\end{proof}

\begin{prop}\label{pp:affine}
For any given two quantum states $\rho$ and $\sigma$ on
$\complex^d$, there exists a unitary operator $U_0$ on $\complex^d$,
which depends on $\rho$ and $\sigma$, such that
\begin{eqnarray}
\rF(\rho,\sigma) = \rA(\rho,U_0\sigma U_0^\dagger),
\end{eqnarray}
where $\rA(\rho,\sigma)$ is called \emph{affine fidelity}
\cite{Luo,Ma}, defined by
$\rA(\rho,\sigma):=\Tr{\sqrt{\rho}\sqrt{\sigma}}$.
\end{prop}

\begin{proof}
Consider a map defined over the unitary group $\rU(\complex^d)$ in
the following:
$$
g(U) = \rA(\rho,U\sigma U^\dagger).
$$
Apparently, $g$ is a continuous map over $\rU(\complex^d)$. Furthermore,
$g(\I_d)\leqslant \rF(\rho,\sigma)$ s a basic matrix inequality. It implies that the affine fidelity is upper bounded by the fidelity. Since the unitary group
$\rU(\complex^d)$ is a compact and connected Lie group, it follows
that the image of $\rU(\complex^d)$ under the map $g$ is a closed
interval. Thus it suffices to show that there exists a unitary
operator $V$ such that $\rF(\rho,\sigma)\leqslant g(V)$. We proceed
with the following result obtained in \cite{Zhang}: there exists a
unitary operator $V\in\rU(\complex^d)$ such that
$$
\rF(\rho,\sigma) = \Tr{\exp\Pa{\log\sqrt{\rho} +
V\log\sqrt{\sigma}V^\dagger}}.
$$
By Golden-Thompson inequality $\Tr{e^{A+B}}\leqslant \Tr{e^Ae^B}$,
where $A$ and $B$ are Hermitian, we get that
$\rF(\rho,\sigma)\leqslant g(V)$. Now the fidelity
$\rF(\rho,\sigma)\in \im(g)$, the image of $g$. Therefore there
exists a unitary operator $U_0\in\rU(\complex^d)$ satisfying the
property that we want. We are done.
\end{proof}

In the following two subsections, we will respectively derive the
upper and lower bounds of $\rG_{\max}$ and $\rG_{\min}$.

\subsection{Bounds of $\rG_{\max}$} \label{subsec:a}

The monotonicity of the fidelity implies the upper bound
\begin{eqnarray}
G_{\max}(\r,\s) \leqslant \min \bigg(
\max_{U_1}\rF(\rho_1,U_1\sigma_1 U^\dagger_1),
\max_{U_2}\rF(\rho_2,U_2\sigma_2 U^\dagger_2), \max_{U_{12}}
\rF(\r,U_{12}\s U_{12}^\dagger) \bigg).
\end{eqnarray}
This bound is analytically derivable, as we have computed explicitly
$\max_U \rF(\rho, U\sigma U^\dagger)$ and $\min_U \rF(\rho, U\sigma
U^\dagger)$ \cite{Zhang}. This result directly applies to the
computation of $\max_{U_{12}} \rF(\r,U_{12}\s U_{12}^\dagger)$. Next
we obtain a lower bound of $\rG_{\max}$. From a well-known fact in
matrix analysis: $\abs{\tr{A}} \leqslant \tr{\abs{A}}$ for any
matrix $A$, where $\abs{A}=\sqrt{A^\dagger A}$, letting
$A=\sqrt{\rho}\sqrt{\sigma}$ gives rise to
$$
\abs{\Tr{\sqrt{\rho}\sqrt{\sigma}}}\leqslant
\Tr{\abs{\sqrt{\rho}\sqrt{\sigma}}}.
$$
Clearly $\Tr{\sqrt{\rho}\sqrt{\sigma}}$ is a nonnegative real number
and $\rF(\rho,\sigma)=\Tr{\abs{\sqrt{\rho}\sqrt{\sigma}}}$, thus
$\rF(\rho,\sigma)\geqslant \rA(\rho,\sigma)$ for any two states
$\rho,\sigma$, then we have
\begin{eqnarray}
\rG_{\max}(\r,\s)&\geqslant&
\max_{U_i\in\rU(\cH_i):i=1,2} \rA(\r,(U_1\ot
U_2)\s(U_1\ot U_2)^\dagger)\notag\\
&\geqslant& \int_{\rU(d_1)}\int_{\rU(d_2)}\rA(\r,(U_1\ot
U_2)\s(U_1\ot U_2)^\dagger)d\mu(U_1) d\mu(U_2)\notag\\
&=&\frac{\Tr{\sqrt{\r}}\Tr{\sqrt{\s}}}{d_1d_2}.
\end{eqnarray}
We have denoted the uniform normalized Haar measure by $\mu(U)$ over
the unitary group. On the other hand, the inequality \cite{Mosonyi}
\begin{eqnarray}\label{eq:mosonyi}
\rS(\rho||\sigma)\geqslant -2\log \rF(\rho,\sigma)
\end{eqnarray}
where $\rS(\rho||\sigma):=\Tr{\rho(\log\rho-\log\sigma)}$ is the
quantum relative entropy, implies
\begin{eqnarray}
\min_{U_i\in\rU(\cH_i):i=1,2} \rS(\r||(U_1\ot
U_2)\s(U_1\ot U_2)^\dagger) \geqslant -2\log
\rG_{\max}(\r,\s).
\end{eqnarray}
So we have obtained a lower bound of \eqref{eq:max}
\begin{eqnarray}
&&\rG_{\max}(\r,\s) \notag\\
&&\geqslant \max\bigg\{
\frac{\Tr{\sqrt{\r}}\Tr{\sqrt{\s}}}{d_1d_2},
\exp\bigg(-{1\over2} \min_{U_i\in\rU(\cH_i):i=1,2}
\rS(\r||(U_1\ot U_2)\s(U_1\ot U_2)^\dagger) \bigg)
\bigg\}.
\end{eqnarray}

\subsection{Bounds of $\rG_{\min}$} \label{subsec:b}

Clearly
\begin{eqnarray}
\rG_{\min}(\r,\s) &\leqslant&
\min_{U_2}\int_{\rU(d_1)}\rF(\r,(U_1\ot
U_2)\s(U_1\ot U_2)^\dagger)d\mu(U_1) \notag\\
&\leqslant& \min_{U_2}\rF(\r,\I_{d_1}/{d_1}\ot U_2\sigma_2
U^\dagger_2)\notag
\\
&\leqslant& \frac1{\sqrt{d_1d_2}}\Tr{\sqrt{\r}},
\end{eqnarray}
where the last inequality follows from the integral over $\rU(d_2)$.
By exchanging $\r$ and $\s$ in the inequality, we
obtain an upper bound of \eqref{eq:min}
\begin{eqnarray}
\rG_{\min}(\r,\s)\leqslant \frac1{\sqrt{d_1d_2}}\min
\Set{\Tr{\sqrt{\r}},\Tr{\sqrt{\s}}}.
\end{eqnarray}

Next, we study the lower bound of $\rG_{\min}$. Let $\rS(\rho):=-\Tr
{\rho \log \rho}$ be the von Neumann entropy with the natural
logarithm $\log$ and $0\log0\equiv0$. It is obtained in \cite{LZ}
that for full-ranked states $\rho,\sigma\in\density{\complex^d}$, we
have
\begin{eqnarray}\label{eq:fullrank}
\rF(\rho,\sigma) \geqslant \Tr{\sqrt{\rho}} \times \exp\Pa{\frac12
\sum_j\lambda^\downarrow_j(\rho) \log\lambda^\uparrow_j(\sigma)},
\end{eqnarray}
where $\lambda^\downarrow(\rho)$ denotes the set of eigenvalues of
$\rho$ in the decreasing order and $\lambda^\uparrow(\sigma)$
denotes the set of eigenvalues of $\sigma$ in the increasing order.
It is known that the fidelity is unchanged under the exchange of
arguments. Assuming $\rho=\r$, $\sigma=(U_1\ot
U_2)\s(U_1\ot U_2)^\dagger$  and exchanging them in
\eqref{eq:fullrank}, we obtain a constant lower bound of
\eqref{eq:min}
\begin{eqnarray}
\rG_{\min}(\r,\s) &\geqslant& \max\bigg\{
\Tr{\sqrt{\r}} \times\exp\Pa{\frac12
\sum_j\lambda^\downarrow_j(\r)
\log\lambda^\uparrow_j(\s)},
\notag\\
&&~~~~~~\Tr{\sqrt{\s}} \times\exp\Pa{\frac12
\sum_j\lambda^\downarrow_j(\s)
\log\lambda^\uparrow_j(\r)},
\notag\\
&&~~~~~~\exp\Pa{-{1\over2}\max_{U_i\in\rU(\cH_i):i=1,2}
\rS(\r||(U_1\ot U_2)\s(U_1\ot U_2)^\dagger)},
\bigg\}
\end{eqnarray}
where the last argument follows from \eqref{eq:mosonyi}.

\section{Discussion}\label{sect:examples}

In this section we investigate the power of local unitaries for the
commutativity of quantum states, the quantification of the
commutativity, and the equivalence of the two optimization problems.
They generalize the previous discussion.

Commutative quantum states can be prepared in the same pure state
basis. They not only share operational mathematical properties, and
also can save resources in experiments. One might expect that, under
local unitary dynamics we could make two non-commutative quantum
states become commutative. That is, given two mixed states
$\r$ and $\s$ with $[\r,\s]\neq0$,
are always there local unitaries $U_1$ and $U_2$ such that
$[\r,(U_1\ot U_2)\s (U_1\ot U_2)^\dagger]=0$?
Unfortunately The answer is negative. Indeed, let $\r=\sum_i
p_i\proj{a_i}$ and $\s=\sum_j q_j\proj{b_j}$ be their
respective spectral decomposition, and $p_i,q_j>0$ for all $i,j$. If
the answer is yes, then there exist two sets $\set{\ket{a_i}}$ and
$\set{\ket{b_j}}$, such that up to overall phases they are from the
same o. n. basis of $\cH_1\ot\cH_2$. Below we construct a
counterexample to this statement. Hence the answer is no. We
consider two two-qubit states
\begin{eqnarray}
&&\r={1\over2}\proj{\Psi^+}+{1\over3}\proj{\Psi^-}+{1\over6}\proj{\Phi^+},
\\
&&\s={2\over3}\proj{00}+{1\over3}\proj{11},
\end{eqnarray}
where $\ket{\Psi^{\pm}}={1\over\sqrt2}(\ket{00}\pm\ket{11})$ and
$\ket{\Phi^{\pm}}={1\over\sqrt2}(\ket{01}\pm\ket{10})$ are the
standard EPR pairs. Since the positive eigenvalues all have
multiplicity one, the eigenstates of them of $\r$ and $\s$
are respectively equal to $\{\ket{\Psi^\pm},\{\ket{\Phi^+}\}$ and
$\{\ket{00},\ket{11}\}$, up to overall phases on these states. If
the answer is yes, then there are local unitaries $U_1,U_2$ such
that $\set{\ket{\Psi^\pm},\{\ket{\Phi^+}}$ and $\set{(U_1\ot
U_2)\ket{00},(U_1\ot U_2)\ket{11}}$ are from the same o. n. basis of
$\complex^2\ot\complex^2$. Since the former consists of entangled
states and the latter consists of separable states, they have to be
pairwise orthogonal. It is impossible because the former and latter
respectively span a 3-dimensional and 2-dimensional subspace in
$\complex^2\ot\complex^2$. \qed


Another interesting problem is whether the two optimization problems
are equivalent for general $\rho$ and $\sigma$. The equivalence
would imply the sufficiency of solving only one of them. We propose
to study two related functionals
$$
\max_{U_i\in\rU(\cH_i):i=1,2}
\Tr{\r(U_1\ot U_2)\s(U_1\ot U_2)^\dagger}
$$
and
$$\min_{U_i\in\rU(\cH_i):i=1,2} \Tr{\r(U_1\ot U_2)\s(U_1\ot
U_2)^\dagger},
$$
as they also measure the similarity between mixed states $\r$
and $\s$. The computation of the two functionals is equivalent,
because if we can compute the former for any $\r$ and
$\s$, then we can also compute the latter by replacing
$\s$ by $\I-\s$; and vice versa. So it suffices to compute
$$
\max_{U_i\in\rU(\cH_i):i=1,2} \Tr{\r(U_1\ot U_2)\s(U_1\ot
U_2)^\dagger}.
$$
Next, it is known that $\abs{\tr{U A }} \leqslant
\tr{\sqrt{A^\dg A}}$ for any unitary $U$ and any matrix $A$. We have
\begin{eqnarray}
\max_{U_i\in\rU(\cH_i):i=1,2} \Tr{\r(U_1\ot U_2)\s(U_1\ot
U_2)^\dagger} \leqslant \rG_{\max}(\r^2,\s^2).
\end{eqnarray}
So the two functionals are not only physically, but also
mathematically related to $\rG_{\max}$ and $\rG_{\min}$.

\section{Conclusions}\label{sect:conclusion}

In this paper we have studied two optimization problems that are
related to many quantum-information problems. The problems are
generally solvable by the SDP, and we manged to work out the
analytical formulae for some states. For mixed states we have
constructed many upper and lower bounds of the two functionals. We
have shown that the entanglement of Werner states might be a more
unaccessible quantum correlation than the separability in terms of
the local unitary dynamics. We have investigated the power of local
unitaries for the commutativity of quantum states and the equivalence of the two
optimization problems. Apart from the problems proposed in last
section, studying the relation between the distillability of Werner
states and their distance to the separable states may shed new light
to the distillability problem.

{\it Acknowledgement.---} We thank Marco Piani for pointing out the
last statement in Theorem \ref{thm:rho12} and Ref. \cite{PM}, as
well as a few minor errors in an early version of this paper. LZ is
grateful for financial support from National Natural Science
Foundation of China (No.11301124). LC was supported by the NSF of China (Grant No. 11501024), and the Fundamental Research Funds for the Central Universities (Grant Nos. 30426401 and 30458601).

\section*{Appendix}

\subsection*{Isotropic state} \label{app:iso}

The isotropic state is the convex mixture of a maximally entangled
state and the maximally mixed state:
\begin{eqnarray*}
\rho_{\text{iso}}(\lambda) = \frac{1-\lambda}{d^2-1}\Pa{\I_d\ot\I_d
- \out{\Psi^+}{\Psi^+}} + \lambda\out{\Psi^+}{\Psi^+},
\end{eqnarray*}
where $\lambda\in[0,1]$ and
$\ket{\Psi^+}=\frac1{\sqrt{d}}\sum^d_{j=1}\ket{jj}$. By
Theorem~\ref{thm:rho12}, we have
$$
\max_{\ket{u},\ket{v}} \abs{\iinner{uv}{\Psi^+}}^2 =
\frac1d~\text{and}~\min_{\ket{u},\ket{v}}
\abs{\iinner{uv}{\Psi^+}}^2 = 0.
$$
Thus
\begin{eqnarray}
\Innerm{uv}{\r_{\text{iso}}(\lambda)}{uv} = \frac{1-\lambda}{d^2-1}
+ \frac{d^2\lambda-1}{d^2-1}\abs{\iinner{uv}{\Psi^+}}^2.
\end{eqnarray}
To further characterize the maximum and minimum of this function, we
discuss two subcases.
\\
(1). If $\frac1{d^2}\leqslant\lambda \leqslant1$, then
$\max_{\ket{u},\ket{v}}\Innerm{uv}{\r_{\text{iso}}(\lambda)}{uv} =
\frac{d\lambda+1}{d(d+1)}$ and
$\min_{\ket{u},\ket{v}}\Innerm{uv}{\r_{\text{iso}}(\lambda)}{uv} =
\frac{1-\lambda}{d^2-1}$;\\
(2). If $0\leqslant \lambda<\frac1{d^2}$, then
$\min_{\ket{u},\ket{v}}\Innerm{uv}{\r_{\text{iso}}(\lambda)}{uv} =
\frac{d\lambda+1}{d(d+1)}$ and
$\max_{\ket{u},\ket{v}}\Innerm{uv}{\r_{\text{iso}}(\lambda)}{uv} =
\frac{1-\lambda}{d^2-1}$. In summary, we have
\begin{eqnarray}
\rG_{\max}(\r_{\text{iso}}(\lambda),\out{uv}{uv}) =
\max\Pa{\sqrt{\frac{d\lambda+1}{d(d+1)}},\sqrt{\frac{1-\lambda}{d^2-1}}},
\end{eqnarray}
and
\begin{eqnarray}
\rG_{\min}(\r_{\text{iso}}(\lambda),\out{uv}{uv}) =
\min\Pa{\sqrt{\frac{d\lambda+1}{d(d+1)}},\sqrt{\frac{1-\lambda}{d^2-1}}}.
\end{eqnarray}

\subsection*{Proof of Proposition \ref{pp:max+min}} \label{app:proof}

For any semi-definite positive matrix $X$, the following inequality
is easily derived via the spectral decomposition of $X$:
\begin{eqnarray}
\label{eq:rankxtrx} \sqrt{\rank (X) \cdot \Tr{X}} \geqslant
\Tr{\sqrt{X}} \geqslant \sqrt{\Tr{X}},
\end{eqnarray}
where the first equality holds if and only if $X$ has identical
nonzero eigenvalues, and the second equality holds if and only if
$X$ has rank one. Let $X=A^{1/2}BA^{1/2}$ for any two semi-definite
positive matrices $A,B$. Then \eqref{eq:rankxtrx} implies
\begin{eqnarray}
\label{eq:trsqrtx} \sqrt{\rank (A^{1/2}B^{1/2}) \cdot
\tr{AB}}\geqslant \rF(A,B)\geqslant\sqrt{\tr{AB}}, ~~~~~~~~~
\end{eqnarray}
where the first equality holds if and only if $A^{1/2}BA^{1/2}$ has
identical nonzero eigenvalues, and the second equality holds if and
only if $A^{1/2}BA^{1/2}$ has rank one. To prove the first
inequality in \eqref{eq:max+min}, let $W_1\ot W_2 = \arg
\rG_{\max}(\r,\s)$. We have
\begin{eqnarray}\label{eq:max+min2'}
&& \rG_{\max}(\r,\s)^2 +
(d_1d_2-1)\rG_{\min}(\r,\s')^2
\notag\\
&&\leqslant  \rF(\r,(W_1\ot W_2)\s(W_1\ot
W_2)^\dagger)^2 + (d_1d_2-1) \rF(\r,(W_1\ot
W_2)\s'(W_1\ot W_2)^\dagger)^2
\notag\\
&& \leqslant  \rank(\r^{1/2} (W_1\ot
W_2)\s^{1/2}(W_1\ot W_2)^\dagger) \Tr{\r(W_1\ot
W_2)\s(W_1\ot W_2)^\dagger}
\notag\\
&&~~~+  \rank(\r^{1/2}(W_1\ot W_2)(\s')^{1/2}(W_1\ot
W_2)^\dagger)(d_1d_2-1)\Tr{\r(W_1\ot W_2)\s'(W_1\ot
W_2)^\dagger}
\notag\\
&& \leqslant \rank(\r) [\Tr{\r(W_1\ot
W_2)\s(W_1\ot W_2)^\dagger}+  (d_1d_2-1)
\tr{\r(W_1\ot W_2)\s'(W_1\ot W_2)^\dagger}]
\notag\\
&&=\rank(\r),
\end{eqnarray}
where the first inequality follows from the definition of
$\rG_{\max}$ and $\rG_{\min}$, and its equality is equivalent to
condition~\eqref{1}. The second inequality in \eqref{eq:max+min2'}
follows from the first inequality in \eqref{eq:trsqrtx} by assuming
$A=\r$, $B=(W_1\ot W_2)\s(W_1\ot W_2)^\dagger$ and
$(W_1\ot W_2)\s'(W_1\ot W_2)^\dagger$, respectively. Its
equality is equivalent to condition~\eqref{2} by the first
inequality in \eqref{eq:trsqrtx}. The third inequality in
\eqref{eq:max+min2'} follows from the fact $\rank (A) \geqslant
\rank (A^{1/2} B^{1/2})$. Its equality holds if condition~\eqref{3}
holds. So we have proved the first inequality, and the three
conditions by which the equality holds in \eqref{eq:max+min}.

To prove the second inequality in \eqref{eq:max+min}, let $V_1\ot
V_2 = \arg \rG_{\min}(\r,\s')$. We have
\begin{eqnarray}\label{eq:max+min2}
&& \rG_{\max}(\r,\s)^2 +
(d_1d_2-1)\rG_{\min}(\r,\s')^2
\notag\\
&&\geqslant \rF(\r,(V_1\ot V_2)\s(V_1\ot
V_2)^\dagger)^2 + (d_1d_2-1)\rF(\r,(V_1\ot
V_2)\s'(V_1\ot V_2)^\dagger)^2
\notag\\
&&\geqslant \tr{\r(V_1\ot V_2)\s(V_1\ot
V_2)^\dagger} + (d_1d_2-1) \tr{\r(V_1\ot
V_2)\s'(V_1\ot V_2)^\dagger}
\notag\\
&&=1,
\end{eqnarray}
where the second inequality follows from \eqref{eq:trsqrtx} by
assuming $A=\r$, $B=(V_1\ot V_2)\s(V_1\ot
V_2)^\dagger$ and $(V_1\ot V_2)\s'(V_1\ot V_2)^\dagger$,
respectively. So we have proved the second inequality in
\eqref{eq:max+min}. The equality in \eqref{eq:max+min} holds if and
only if the first two equalities in \eqref{eq:max+min2} both hold.
The first equality is equivalent to condition~\eqref{4} by the
definition of $\rG_{\max}$ and $\rG_{\min}$, and the second equality
is equivalent to conditions~\eqref{5} and \eqref{6} by
\eqref{eq:trsqrtx}. This completes the proof.



\begin{thebibliography}{999}


\bibitem{Watrous}
The vector-operator correspondence
$\vec\Pa{\sum_{i,j}X_{ij}\out{i}{j}} := \sum_{i,j} X_{ij}\ket{ij}$
is defined, e.g., in J. Watrous, Theory of Quantum Information,
University of Waterloo, Waterloo (2008). See
\url{http://www.cs.uwaterloo.ca/~watrous/quant-info/}

\bibitem{scz09}
W. Song, L. Chen, and S.-L. Zhu,
\pra~\href{http://dx.doi.org/10.1103/PhysRevA.80.012331}{\textbf{80},
012331 (2009).}

\bibitem{SK-jpa}
A. Sawicki, M. Ku\'{s}, Geometry of the local equivalence of states.
\jpa: Math. Theor.
\href{http://dx.doi.org/10.1088/1751-8113/49/44/495301}{\textbf{44}(49),
495301 (2011).}

\bibitem{HKS-jmp}
A. Huckleberry, M. Ku\'{s}, and A. Sawicki, Bipartite entanglement,
spherical actions, and geometry of local unitary orbits,
\jmp~\href{http://dx.doi.org/10.1063/1.4791681}{\textbf{54}, 022202
(2013).}

\bibitem{MOS-jmp}
T. Maciazek, M. Oszmaniec, and A. Sawicki, How many invariant
polynomials are needed to decide local unitary equivalence of qubit
states? \jmp~\href{http://dx.doi.org/10.1063/1.4819499}{\textbf{54},
092201 (2013).}

\bibitem{PMGDHZ-jpa}
Z. Puchala, J.A. Miszczak, P. Gawron, F. Dunkl, J.A. Holbrook, and
K. \.{Z}yczkowski, Restricted numerical shadow and the geometry of
quantum entanglement, \jpa: Math. Theor.
\href{http://dx.doi.org/10.1088/1751-8113/45/41/415309}{\textbf{41}(45),
415309 (2012).}

\bibitem{SK-pra}
A. Sawicki, M. Oszmaniec, and M. Ku\'{s}, Critical sets of the total
variance can detect all stochastic local operations and classical
communication classes of multiparticle entanglement,
\pra~\href{http://dx.doi.org/10.1103/PhysRevA.86.040304}{\textbf{86},
040304 (2012).}

\bibitem{GW-prl}
G. Gour and N.R. Wallach, Classification of Multipartite
Entanglement of All Finite Dimensionality,
\prl~\href{http://dx.doi.org/10.1103/PhysRevLett.111.060502}{\textbf{111},
060502 (2011).}



\bibitem{Jevtic}
S. Jevtic, D. Jennings, and T. Rudolph,
\prl~\href{http://dx.doi.org/10.1103/PhysRevLett.108.110403}{\textbf{108},
110403 (2012).}
\pra~\href{http://dx.doi.org/10.1103/PhysRevA.85.052121}{\textbf{85},
052121 (2012).}

\bibitem{Modi}
K. Modi and M. Gu, Int. J. Mod. Phys. B
\href{http://dx.doi.org/10.1142/S0217979213450276}{\textbf{27},
1345027 (2013).}

\bibitem{bcz10}
Stefanie Barz, Gunther Cronenberg, Anton Zeilinger, and Philip
Walther, Nature Photonics {\bf4}, 553 (2010).

\bibitem{sbm11}
Philipp Schindler, Julio T. Barreiro, Thomas Monz, Volckmar
Nebendahl, Daniel Nigg, Michael Chwalla, Markus Hennrich, Rainer
Blatt, Science, {\bf332}, 1059 (2011).

\bibitem{Zhang}
L. Zhang and S.-M. Fei, \jpa: Theor. Math.
\href{http://dx.doi.org/10.1088/1751-8113/47/5/055301}{\textbf{47},
055301 (2014).}

\bibitem{shimony95}
A. Shimony, Ann. N. Y. Acad. Sci.
\href{http://dx.doi.org/10.1111/j.1749-6632.1995.tb39008.x}{\textbf{755},
675 (1995).}

\bibitem{wg03}
T-C. Wei and P.M. Goldbart,
\pra~\href{http://dx.doi.org/10.1103/PhysRevA.68.042307}{\textbf{68},
042307 (2003).}

\bibitem{zch10}
H. Zhu, L. Chen, and M. Hayashi,
\njp~\href{http://dx.doi.org/10.1088/1367-2630/12/8/083002}{\textbf{12},
083002 (2010).}

\bibitem{cah14}
L. Chen, M. Aulbach and M. Hajdusek,
\pra~\href{http://dx.doi.org/10.1103/PhysRevA.89.042305}{\textbf{89},
042305 (2014).}

\bibitem{ssb04}
Y. Shimoni, D. Shapira, and O. Biham,
\pra~\href{http://dx.doi.org/10.1103/PhysRevA.69.062303}{\textbf{69},
062303 (2004).}

\bibitem{bno02}
O. Biham, M. A. Nielsen, and T. J. Osborne,
\pra~\href{http://dx.doi.org/10.1103/PhysRevA.65.062312}{\textbf{65},
062312 (2002).}

\bibitem{mmv07}
D. Markham, A. Miyake, and S. Virmani,
\njp~\href{http://dx.doi.org/10.1088/1367-2630/9/6/194}{\textbf{9},
194 (2007).}

\bibitem{Zhao}
M.-J. Zhao,
\pra~\href{http://dx.doi.org/10.1103/PhysRevA.91.012310}{\textbf{91},
012310 (2015).}

\bibitem{Bennett}
C.H. Bennett, D.P. DiVincenzo, J.A. Smolin, and W.K. Wootters,
\pra~\href{http://dx.doi.org/10.1103/PhysRevA.54.3824}{\textbf{54},
3824 (1996).}

\bibitem{GEJ}
J. Grondalski, D. M. Etlinger, and D. F. V. James,
\pla~\href{http://dx.doi.org/10.1016/S0375-9601(02)00884-8}{\textbf{300},
573 (2002).}

\bibitem{sst03} Peter W. Shor, John A. Smolin, and Ashish V. Thapliyal, \prl, {\bf90}, 107901 (2003).

\bibitem{dss00}
Here $\s^\Gamma=\sum_{i,j} \bra{i}\s\ket{j}\otimes \ketbra{j}{i}$ denotes the partial transpose of any state $\s$, where $\{\ket{i}\}$ is the computational basis of space $\cH_A$
\cite{peres96}. See more about the distillability problem in, D. P.
DiVincenzo, P. W. Shor, J. A. Smolin, B. M. Terhal, and A. V.
Thapliyal,
\pra~\href{http://dx.doi.org/10.1103/PhysRevA.61.062312}{\textbf{61},
062312 (2000).}

\bibitem{peres96}
A. Peres,
\prl~\href{http://dx.doi.org/10.1103/PhysRevLett.77.1413}{\textbf{77},
1413 (1996).}

\bibitem{cd11}
L. Chen and D.Z. Djokovic, \jpa: Math. Theor.
\href{http://dx.doi.org/10.1088/1751-8113/44/28/285303}{\textbf{44},
285303 (2011).}

\bibitem{rains01}
Eric M. Rains, arXiv:quant-ph/0008047 (2001).

\bibitem{dps02}
A. C. Doherty, Pablo A. Parrilo, and Federico M. Spedalieri, \prl
{\bf88}, 187904 (2002).

\bibitem{brandao05}
FGSL Brandao, \pra {\bf72}, 022310 (2005).

\bibitem{hhh09}
Horodecki, R. Horodecki, P. Horodecki, M. Horodecki, K, \rmp
{\bf81}, 865 (2009).

\bibitem{JW}
J. Watrous, arXiv:1207.5726

\bibitem{Kill}
N. Killoran, Entanglement quantification and quantum benchmarking of
optical communication devices, PhD thesis, University of Waterloo,
(2012).


\bibitem{Bhatia}
Here $s_k(X)$ stands for the singular values of matrix $X$. See R.
Bhatia, Matrix Analysis, Springer-Verlag (1997).

\bibitem{PM}
M. Piani, C. Mora.
\pra~\href{http://dx.doi.org/10.1103/PhysRevA.75.012305}{\textbf{75},
012305 (2007).}

\bibitem{Johnston}
N. Johnston and D.W.Kribs,
\jmp~\href{http://dx.doi.org/10.1063/1.3459068}{\textbf{51}, 082202
(2010).}

\bibitem{Kribs}
N. Johnston and D.W.Kribs, Quant. Inf. Comput. \textbf{11}(1-2),
0104-0123 (2011).

\bibitem{vd06}
Reinaldo O. Vianna and Andrew C. Doherty, \pra {\bf74}, 052306
(2006).

\bibitem{Coles}
P.C. Coles, J. Kaniewski, and S. Wehner, Nat. Commun.
\href{http://dx.doi.org/10.1038/ncomms6814}{\textbf{5}, 5814
(2014).}

\bibitem{AR}
A.E. Rastegin, J. Opt. B: Quantum Semiclass. Opt.
\href{http://dx.doi.org/10.1088/1464-4266/5/6/017}{\textbf{5}, S647
(2003).}

\bibitem{Roga}
M. Fannes, F.D. Melo, W. Roga, and K. \.{Z}yczkowski, Qaunt. Inf.
Comput. \textbf{12}(5-6), 472-489 (2012).

\bibitem{Fawzi}
O. Fawzi and R. Renner,
\href{http://arxiv.org/abs/1410.0664}{arXiv:1410.0664}.

\bibitem{Luo}
S. Luo and Q. Zhang,
\pra~\href{http://dx.doi.org/10.1103/PhysRevA.69.032106}{\textbf{69},
032106 (2004).}

\bibitem{Ma}
Z.-H. Ma, F.-L. Zhang, and J.-L. Chen,
\pra~\href{http://dx.doi.org/10.1103/PhysRevA.78.064305}{\textbf{78},
064305 (2008).}

\bibitem{Mosonyi}
M. Mosonyi, F. Hiai, IEEE Trans. Inf. Theor. {\bf57}, 2474-2487
(2011).

\bibitem{LZ}
L. Zhang and J. Wu, \jpa: Math. Theor.
\href{http://dx.doi.org/10.1088/1751-8113/47/41/415303}{\textbf{47},
 415303 (2014).}

\end{thebibliography}
\end{document}